\documentclass[12pt]{article}

\usepackage{graphics,graphicx,fullpage,natbib,multirow}
\usepackage{amsmath,amssymb,verbatim,epsfig}
\usepackage[dvipsnames,usenames]{color}

\newtheorem{proposition}{Proposition}

\newtheorem{defin}{\bf Definition}

\newenvironment{proof}{\noindent{\bf Proof}}{$\diamond$}
\newenvironment{proof2}{\noindent{\bf Proof}}{}

\def\ga{\mbox{Ga}}
\def\iga{\mbox{IGa}}

\def\be{\mbox{Be}}
\def\ibe{\mbox{IBe}}

\def\bin{\mbox{Bin}}

\def\ghs{\mbox{GHS}}
\def\no{\mbox{N}}
\def\nbi{\mbox{NB}}

\def\po{\mbox{Po}}

\def\gsst{\mbox{GSSt}}
\def\un{\mbox{Un}}

\def\E{\mbox{E}}
\def\V{\mbox{Var}}
\def\Cov{\mbox{Cov}}
\def\Cor{\mbox{Corr}}

\def\d{\mbox{d}}
\def\rest{\mbox{rest}}

\def\bc{{\bf c}}

\def\bs{{\bf s}}

\def\bw{{\bf w}}
\def\bx{{\bf x}}
\def\by{{\bf y}}

\def\bC{{\bf C}}

\def\bM{{\bf M}}

\def\bS{{\bf S}}

\def\bY{{\bf Y}}
\def\bZ{{\bf Z}}

\def\bone{{\bf 1}}
\def\simind{\stackrel{\mbox{\scriptsize{ind}}}{\sim}}

\newcommand{\bfeta}{\boldsymbol{\eta}}

\newcommand{\GC}{\mathcal{G}}

\newcommand{\MC}{\mathcal{M}}

\newcommand{\RB}{\mathbb{R}}

\begin{document}

\baselineskip=24pt

\title{\bf Graphical models with marginals in the exponential family }
\author{{\sc Luis E. Nieto-Barajas \& Sim\'on Lunag\'omez} \\[2mm]
{\sl Department of Statistics, ITAM, Mexico} \\[2mm]
{\small {\tt luis.nieto@itam.mx \& simon.lunagomez@itam.mx}} \\}
\date{}
\maketitle

\begin{abstract}
Graphical models encode conditional independence statements of a multivariate distribution via a graph. Traditionally, the marginal distributions in a graphical model are assumed to be Gaussian. In this paper, we propose a three-level hierarchical model that functions as the hyper-Markov law that enables a graphical model with marginals in the exponential family with quadratic variance function. Inference on the model parameters is made using a Bayesian approach. We perform a simulation study and real data analyses to illustrate the usefulness of our models. 
\end{abstract}

\vspace{0.2in} \noindent {\sl Keywords}: Bayesian inference, exponential family, hierarchical models, graphical models, multivariate normal.

\section{Introduction}
\label{sec:intro}

The formulation of graphical models was originally proposed by \cite{giudici&green:99}. This consists of a hierarchical model of the form
\begin{equation}
\label{eq:ggm}
f(\bx, \Sigma, \GC) = f(\bx\mid \Sigma, \GC) f(\Sigma \mid \GC) f(\GC),
\end{equation}
where $f(\bx\mid\Sigma,\GC)$ is the likelihood, $f(\Sigma \mid \mathcal{G})$ is the hyper-Markov law (HML) and $f(\GC)$ is the hyper-prior on graph space. 
Probabilistic graphical models are usually characterised by a multivariate Gaussian distribution on random variables associated to each of the nodes, $f(\bx\mid\Sigma,\GC)$. The association among variables (nodes) is described through the variance covariance matrix $\Sigma$. The inverse of this matrix is known as precision matrix, where a zero entry at position $(i,k)$ induces conditional independence between the variables $j$ and $k$ and, therefore, a missing edge in the graph between the nodes $j$ and $k$. 

Estimating a precision matrix is not trivial since it has to be symmetric and positive definite. Most contributions in the area of graphical models have concentrated in proposing a way of estimating the precision matrix. In a Bayesian context, \cite{khare&al:18} proposed a generalised G-Wishart distribution. For high-dimensional graphs, \cite{gan&al:19} induce shrinkage and sparsity with a mixture of Laplace priors. The use of covariates has been studied by \cite{chen&al:16} who use them to adjust the location, and \cite{ni&al:22} who use them in the estimation of the precision matrix. 

As discussed in \cite{giudici&green:99}, the hardest component to generalise or modify in \eqref{eq:ggm} is the HML. A first attempt was made by \cite{pitt&al:06} who use a Gaussian copula regression model. In this paper, we consider a hierarchical model, originally conceived for spatio-temporal data, and turn it into a hyper-Markov Gaussian law. Based on a  similar hierarchical model, we go beyond Gaussianity and use members of the exponential family with quadratic variance to produce a new family of HML for graphical models. One of these members is the normal model. We additionally simplify the precision matrix estimation by constraining the dependence between variables to be non-negative. The exponential family HML is based on a three-level hierarchical model that guarantees that the marginal distribution of the nodes is invariant and the graphical dependence is driven by a latent variable that links each pair of nodes. 

The rest of the paper is as follows: In Section \ref{sec:gmodel} we re-write the Gaussian graphical model using a three-level hierarchical representation and study its properties. In Section \ref{sec:emodel} we rely on a similar hierarchical model to extend the HML to the exponential family with quadratic variance. We also present the six members of the family. Bayesian inference is presented in Section \ref{sec:bayes} and numerical studies are included in Section \ref{sec:app}, where we include simulation and real data studies. Section \ref{sec:disc} concludes with a discussion. 

Before we proceed, let us introduce some notation. $\po(\lambda)$ denotes a Poisson density with mean $\lambda$; $\bin(m,p)$ denotes a binomial density with number of trials $m$ and probability of success $p$; $\nbi(m,p)$ denotes a negative binomial density with number of failures $m$ and success probability $p$; $\un(a,b)$ denotes a uniform density in the interval $(a,b)$; $\ga(\alpha,\beta)$ denotes a gamma density with mean $\alpha/\beta$; $\iga(\alpha,\beta)$ denotes a inverse gamma density with mean $\beta/(\alpha-1)$; $\be(\alpha,\beta)$ denotes a beta density with mean $\alpha/(\alpha+\beta)$; $\ibe(\alpha,\beta)$ denotes an inverse beta density with mean $\alpha/(\beta-1)$; $\ghs(\mu,\alpha)$ denotes a generalised hyperbolic secant density with mean $\mu$ and precision $\alpha$ \citep{morris:82}; $\gsst(\mu,m)$ denotes a generalised scaled student t density with mean $\mu$ and precision $m$ \citep{morris:83}; $\no(m,c)$ denotes a univariate normal density with mean $m$ and precision $c$; $\no_p(\bM,\bC)$ denotes a $p$-variate normal density with mean vector $\bM_{p\times 1}$ and precision matrix $\bC_{p\times p}$;. The density evaluated at a specific point $x$, will be denoted, for instance for the normal case, as $\no(x\mid m,c)$.

\section{HML for the Gaussian model}
\label{sec:gmodel}

Recently, \cite{nieto:20} introduced a spatio-temporal dependence model with identically distributed normal marginal distributions and joint multivariate normal distribution. For each node $i$, they defined the variable of interest $Y_j$ and a latent variable $Z_j$, together with a common variable $W$ for all nodes. Their construction is based on a three-level hierarchical model, where the dependence among variables $\{Y_j\}$ is induced by selecting any subset of $\{Z_j\}$ and exploiting conjugacy properties of the normal-normal Bayesian model. 

Unlike \cite{nieto:20}, here we define a set of latent variables $\{Z_{i,j}\}$ that will play the role of weights for the edges between any pair of nodes $(i,j)$. In other words, these parameters encode the strenght of the dependence between variables given the conditional independence statements provided by the graph, which means that the hierarchical model works as the hyper-Markov law. To illustrate, let us consider a scenario with $n=3$ nodes as the one depicted in Figure \ref{fig:net3}. 

In principle, we allow all nodes $Y_j$, $j=1,\ldots,n$ to be linked via $Z_{jk}$, $j\neq k=1,\ldots,p$. In Figure \ref{fig:net3} each node $Y_j$ is linked to the other two nodes via the latent variables $Z_{jk}$ for $j\neq k=1,2,3$. In general, the new hierarchical construction for the three sets of variables $\bY=\{Y_j\}$, $\bZ=\{Z_{j,k}\}$ and $W$ is given by 
\begin{align}
\label{eq:normal}
\nonumber
W&\sim\no(m_0,c_0) \\
Z_{j,k}\mid W&\simind\no(w,c_{jk}) \\
\nonumber
Y_j\mid \bZ&\simind\no\left(\frac{c_0 m_0+\sum_{k\neq j}^n c_{jk}z_{jk}}{c_0+\sum_{k\neq j}^n c_{jk}},\,c_0+\sum_{k\neq j}^n c_{jk}\right),
\end{align}
for $j\neq k=1,\ldots,p$, where $Z_{jk}=Z_{kj}$ with probability one (w.p.1) and $c_{jk}\equiv c_{kj}$, so the number of elements in $\bZ$, and in $\bc=\{c_{jk}\}$, is only $n(n-1)/2$. The parameters of the model satisfy $m_0\in\RB$, $c_0>0$ and $c_{jk}\geq 0$. To illustrate, Figure \ref{fig:graphm} shows the hierarchical model of three levels \eqref{eq:normal} associated with Figure \ref{fig:net3}. 
The properties of the marginal distributions of $Y_j$'s and the dependence induced among $\bY$ are given in Proposition \ref{prop:normal}. 

\begin{proposition} 
\label{prop:normal}
Let $\bY$ be an $n$-dimensional vector whose joint distribution is defined by equations \eqref{eq:normal}. After marginalising the latent variables $\bZ$ and $W$, the joint distribution of $\bY$ is a multivariate normal, in notation $\bY\sim\no_n(\bM_0,\bC)$, with mean vector $\bM_0=m_0\bone$ and $\bone$ a vector of ones of dimension $n\times 1$ and precision matrix $\bC$ of dimension $n\times n$, with variance-covariance matrix $\bC^{-1}$ defined by variances $\V(Y_j)=1/c_0$ and correlations between $Y_j$ and $Y_k$ given by 
\begin{equation}
\label{eq:cor}
\Cor(Y_j,Y_k)=\frac{c_0 c_{jk}+\left(\sum_{l\neq j}^n c_{jl}\right)\left(\sum_{l\neq k}^n c_{kl}\right)}{\left(c_0+\sum_{l\neq j}^n c_{jl}\right)\left(c_0+\sum_{l\neq k}^n c_{kl}\right)},
\end{equation}
for $j\neq k=1,\ldots,p$. Moreover, the marginal distribution for each node is invariant and is the same as the latent $W$, that is, $Y_j\sim\no(m_0,c_0)$ for $j=1,\ldots,p$. 
\end{proposition}
\begin{proof}
The $Z_{jk}$'s are conditional independent normals, given $W$, and $W$ is univariate normal, so after marginalising with respect to $W$, the joint distribution of $\bZ$ is multivariate normal with mean, variance and covariances obtained using iterative formulae as $\E(Z_{jk})=m_0$, $\V(Z_{jk})=(c_0+c_{jk})/(c_0c_{jk})$ and $\Cov(Z_{jk},Z_{j'k'})=1/c_0$ for $j\neq k \neq j'\neq k'$. Now, conditional on $\bZ$, the $Y_i$'s are independent, so $\bY\mid\bZ$ is multivariate normal, and since $\bZ$ is multivariate normal, the marginal distribution of $\bY$ is again multivariate normal. To obtain its parameters we first get the marginal distributions of each $Y_j$. Considering the third equation in \eqref{eq:normal}, the conditional distribution of $Y_j$ given $\bZ$ depends on $\sum_{k\neq j}c_{jk}Z_{jk}$ which conditionally on $W$ has $\no(w\sum_{k\neq j}c_{jk},1/\sum_{k\neq j}c_{jk})$ distribution. Integrating with respect to $W$, the marginal distribution of $\sum_{k\neq j}c_{jk}Z_{jk}$ is $\no(m_0\sum_{k\neq j}c_{jk},c_0/\{(\sum_{k\neq j}c_{jk})(c_0+\sum_{k\neq j}c_{jk})\})$. Relying
on conjugacy properties of the normal model, we obtain that $Y_j\sim\no(m_0,c_0)$ for $j=1,\ldots,p$. For the correlation we first obtain the covariance using iterative formulae. $\Cov(Y_j,Y_k)=\E\{\Cov(Y_j,Y_k\mid\bZ)\}+\Cov\{\E(Y_j\mid\bZ),\E(Y_k\mid\bZ)\}$. Since the $Y_j$'s are conditionally independent given $\bZ$, the first term is zero, and the second term after removing the constants becomes $\Cov(\sum_{l\neq j}c_{jl}Z_{jl},\sum_{l\neq k}c_{kl}Z_{kl})/\{(c_0+\sum_{l\neq j}c_{jl})(c_0+\sum_{l\neq k}c_{kl})\}$. Applying the iterative formulae for a second time, the numerator is $\E\{\Cov(\sum_{l\neq j}c_{jl}Z_{jl},\sum_{l\neq k}c_{kl}Z_{kl}\mid W)\}+\Cov\{\E(\sum_{l\neq j}c_{jl}Z_{jl}\mid W),\E(\sum_{l\neq k}c_{kl}Z_{kl}\mid W)\}$. The first term, after separating the sums in the common part, reduces to $\E\{\V(c_{jk}Z_{jk}\mid W)\}=c_{jk}$. The second term reduces to $(\sum_{l\neq j}c_{jl})(\sum_{l\neq k}c_{kl})\V(W)$. Since the $Y_j$'s and $W$ have the same distribution, dividing by the product of the standard deviations, we obtain the result.
\end{proof}

We note that the conditional distribution of each node $Y_j$ in \eqref{eq:normal}, depends on the links $Z_{jk}$ only through the product $c_{jk}Z_{jk}$, so if $c_{jk}=0$ the contribution of the corresponding link $Z_{jk}$ disappears. In the example depicted in Figure \ref{fig:net3}, if we make $c_{12}=0$ then $Y_1$ and $Y_2$ will only be linked through $Y_3$ (via $Z_{13}$ and $Z_{23}$). To fully understand what is going on, let us consider the variance-covariance matrix defined as 
$$\bC^{-1}=\frac{1}{c_0}\left(\begin{array}{ccc}
1 & \rho_{12} & \rho_{13} \\
\rho_{12} & 1 & \rho_{23} \\
\rho_{13} & \rho_{23} & 1 
\end{array}\right)$$ 
where correlations are given in Proposition \ref{prop:normal} and have the explicit form
\begin{eqnarray}
\nonumber
\rho_{12}=\{c_0c_{12}+(c_{12}+c_{13})(c_{12}+c_{23})\}/\{(c_0+c_{12}+c_{13})(c_0+c_{12}+c_{23})\}, \\ 
\label{eq:rhos}
\rho_{13}=\{c_0c_{13}+(c_{12}+c_{13})(c_{13}+c_{23})\}/\{(c_0+c_{12}+c_{13})(c_0+c_{13}+c_{23})\}, \\ 
\nonumber
\rho_{23}=\{c_0c_{23}+(c_{12}+c_{23})(c_{13}+c_{23})\}/\{(c_0+c_{12}+c_{23})(c_0+c_{13}+c_{23})\}. 
\end{eqnarray}
To obtain the precision matrix, we first calculate the determinant $\det(\bC^{-1})=D/c_0^3$, where $D=1+2\rho_{12}\rho_{13}\rho_{23}-\rho_{12}^2-\rho_{13}^2-\rho_{23}^2$. Then, we compute the cofactor matrix to obtain 
$$\bC=\frac{c_0}{D}\left(\begin{array}{ccc}
1-\rho_{23}^2 & \rho_{23}\rho_{13}-\rho_{12} & \rho_{12}\rho_{23}-\rho_{13} \\
\rho_{23}\rho_{13}-\rho_{12} & 1-\rho_{13}^2 &  \rho_{12}\rho_{13}-\rho_{23} \\
\rho_{12}\rho_{23}-\rho_{13} & \rho_{12}\rho_{13}-\rho_{23} & 1-\rho_{12}^2
\end{array}\right).$$

Now, if $c_{12}=0$ correlations \eqref{eq:rhos} become $\rho_{12}=c_{13}c_{23}/\{(c_0+c_{13})(c_0+c_{23})\}$, $\rho_{13}=c_{13}/(c_0+c_{13})$ and $\rho_{23}=c_{23}/(c_0+c_{23})$, where clearly $\rho_{12}=\rho_{13}\rho_{23}$. Therefore, the corresponding entry $(1,2)$ of the precision matrix $\bC$ has a value of zero. 

In general, in construction \eqref{eq:normal}, $Y_j$ and $Y_k$ are directly linked only through $Z_{jk}$. Moreover, the conditional distribution of $Y_j$ depends on the link $Z_{jk}$ through $c_{jk}Z_{jk}$. In turn, its conditional distribution, given $W$ is $\no(w c_{jk},1/c_{jk})$, then $\V(c_{jk}Z_{jk}\mid W)=c_{jk}$. When $c_{jk}=0$, this latter variance is zero, and thus $c_{jk}Z_{jk}=0$ w.p.1. Therefore,
$Y_j$ and $Y_k$ become conditionally independent given the rest of the nodes. 

In summary, we have defined an alternative way of writing a multivariate normal graphical model in terms of a three-level hierarchical model. There are several features that are worth mentioning. In a general Gaussian graphical model, the variances for each node can be different and the covariance between nodes can be negative or positive, however construction \eqref{eq:normal} implies that the variances are the same for all nodes and the covariances are non negative. On the other hand, 
with construction \eqref{eq:normal} we can easily implement a Gibbs sampler and do not require the computations of determinants or matrices inversions, which are usually very computationally costly, so our model can be scalable; additionally, construction \eqref{eq:normal} can be generalised to the exponential family with quadratic variance functions to define non-Gaussian graphical models, which are described in the following Section \ref{sec:emodel}.

\section{HML for the exponential family model}
\label{sec:emodel}

The hierarchical construction \eqref{eq:normal} is based on conjugacy properties of the normal-normal Bayesian model. The first distribution corresponds to the prior, the second to the likelihood, and the third to the posterior. There are several other models that admit conjugate distributions in Bayesian analysis. In fact, \cite{diaconis&ylvisaker:79} developed the theory of conjugate families for exponential families. Recently, \cite{nieto&gutierrez:22} generalised the spatio-temporal dependence models of \cite{nieto:20} to exponential families. Here we use their ideas to extend the normal graphical model \eqref{eq:normal} to exponential family graphical models. 

Let $Z_1,Z_2,\ldots,Z_n$ be a sample from an exponential family with density $f(z\mid\theta(\mu))=a(z)\exp\{\theta(\mu)s(z)-M(\theta(\mu))\}$, where $a(\cdot)$ is a nonnegative function, $\theta$ is the canonical parameter written in terms of the mean parameterisation $\mu\in\MC$ with $\MC$ the mean parameter space, $S(Z)$ is the canonical statistic and $M(\theta)$ is the cumulant transform such that $\mu=\E(S\mid\theta)=(\partial/\partial\theta)M(\theta)$ and $V(\mu)=\V(S\mid\theta)=(\partial^2/\partial\theta^2)M(\theta)$, where $V(\mu)$ is called the variance function written in terms of the mean $\mu$. \cite{morris:82} characterised six families whose variance function is quadratic, that is, 
\begin{equation}
\label{eq:qvf}
V(\mu)=\nu_0+\nu_1\mu+\nu_2\mu^2.
\end{equation}

Let $S_n=\sum_{i=1}^n S(Z_i)$ be the sufficient statistic for the exponential family density above, then the likelihood becomes $f(s_n\mid\theta,n)=b(s_n,n)\exp\{\theta(\mu)s_n-nM(\theta(\mu))\}$. The conjugate family for the parameter $\mu\in\MC$ is $f(\mu\mid s_0,n_0)=h(s_0,n_0)\exp\{\theta(\mu)s_0-n_0M(\theta(\mu))\}|J_\theta(\mu)|$, with parameters $s_0$ and $n_0$, where $|J_\theta(\cdot)|$ is the Jacobian of the transformation $\theta(\cdot)$. Moments of this prior have general expressions, in particular, $\E(\mu\mid s_0,n_0)=s_0/n_0$ \citep{diaconis&ylvisaker:79} and $\V(\mu\mid s_0,n_0)=V(s_0/n_0)/(n_0-\nu_2)$ \citep{morris:83}. The posterior distribution for $\mu$ has exactly the same form as the prior but with updated parameters $s^*=s_0+s_n$ and $n^*=n_0+n$. 

We are now in a position to define the general graphical model. Let $\bY=\{Y_j\}$, for $j=1,\ldots,p$ be the set of nodes of interest. Let $\bS=\{S_{jk}\}$ for $j\neq k=1,\ldots,p$ be a set of latent variables that play the role of links for any pair of nodes, such that $S_{kj}=S_{jk}$ w.p.1, and let $W$ be a common variable that plays the role of anchor. Then, the exponential family graphical model, with quadratic variance function, is defined by a three level hierarchical model of the form
\begin{align}
\nonumber
W\sim\; & f(w\mid s_0,c_0)=h(s_0,c_0)\exp\{\theta(w)s_0-c_0M(\theta(w))\}|J_\theta(w)| \\
\label{eq:expf}
S_{jk}\mid W\simind\; & f(s_{jk}\mid\theta(w),c_{jk})=b(s_{jk},c_{jk})\exp\{\theta(w)s_{jk}-c_{jk}M(\theta(w))\} \\
\nonumber
Y_j\mid\bS\simind\; & f(y_j\mid s_j^*,c_j^*)=h(s_j^*,c_j^*)\exp\{\theta(y_j)s_j^*-c_j^*M(\theta(y_j))\}|J_\theta(y_j)|, \\
\nonumber
& \mbox{with}\quad s_j^*=s_0+\sum_{k\neq j}s_{jk}\quad\mbox{and}\quad c_j^*=c_0+\sum_{k\neq j}c_{jk}
\end{align}
Parameters of the model are $s_0$, whose parameter space depends on the specific model, $c_0>0$ and $c_{jk}\geq 0$ with $c_{kj}\equiv c_{jk}$ for $j\neq k=1,\ldots,p$. The properties of the general construction \eqref{eq:expf} are given in the following proposition. 

\begin{proposition}
\label{prop:expf}
Let $\bY=\{Y_j\}$ be a network with $p$ nodes, whose distribution is characterized by equations \eqref{eq:expf}. Then, 
\begin{enumerate}
\item[(a)] The marginal distribution of each $Y_j$ is the same as the distribution of $W$, and therefore 
$\E(Y_j)=s_0/c_0$ and $\V(Y_j)=V(s_0/c_0)/(c_0-\nu_2)$ for all $j=1,\ldots,p$. 
\item[(b)] The correlation induced by the model, for any pair of nodes $(Y_j,Y_k)$ is given by \eqref{eq:cor}.
\item[(c)] If $c_{jk}=0$ then nodes $Y_j$ and $Y_k$ are conditionally independent. 
 \end{enumerate}
\end{proposition}
\begin{proof2}
\begin{enumerate}
\item[(a)] Using Bayes' Theorem, we can express the prior for $\mu$ in terms of the posterior as 
$$f(\mu\mid s_0,n_0)=\int f(\mu\mid s^*,n^*)f(s_n\mid s_0,n_0)\d s_n,$$
where $f(s_n\mid s_0,n_0)=\int f(s_n\mid\theta(\tilde\mu),n)f(\tilde\mu\mid s_0,n_0)\d\tilde\mu$ is the prior predictive for the sufficient statistic $S_n$ and $s^*=s_0+s_n$ and $n^*=n_0+n$. Substituting the latter expression into the former, we get
$$f(\mu\mid s_0,n_0)=\int f(\mu\mid s^*,n^*)\int f(s_n\mid\theta(\tilde\mu),n)f(\tilde\mu\mid s_0,n_0)\d\tilde\mu\,\d s_n.$$
Now replacing $(\mu,\tilde\mu,s^*,n^*,s_n,n,n_0)$ by $(y_j,w,s_j^*,c_j^*,\sum_{k\neq j}s_{jk},\sum_{k\neq j}c_{jk},c_0)$ and using Fubini's Theorem we get
$$f(y_j\mid s_0,c_0)=\int\int f(y_j\mid s_j^*,c_j^*) f(\sum_{k\neq j}s_{jk}\mid\theta(w),\sum_{k\neq j}c_{jk})f(w\mid s_0,c_0)\d w\,\d(\sum_{k\neq j}s_{jk}),$$
which is the marginal distribution of $Y_j$ implied by construction \eqref{eq:expf}. This proves that the marginal distribution of $Y_j$ is the same as that of $W$. The mean and variance follow from the exponential family properties. 
\item[(b)] We first recall that $\E(W)=s_0/c_0$, $\V(W)=V(s_0/c_0)/(c_0-\nu_2)$, $\E(S_{jk}\mid W)=c_{jk}w$, $\V(S_{jk}\mid W)=c_{jk}V(W)$, $\E(Y_j\mid\bS)=s_j^*/c_j^*$ and $\V(Y_j)=V(s_j^*/c_j^*)/(c_j^*-\nu_2)$, where the quadratic variance function $V(\cdot)$ is given in \eqref{eq:qvf}. Now, we compute the covariance between $Y_j$ and $Y_k$ using iterative formulae. $\Cov(Y_j,Y_k)=\E\{\Cov(Y_j,Y_k\mid\bS)\}+\Cov\{\E(Y_j\mid\bS),\E(Y_k\mid\bS)\}$. Since the $Y_j$'s are conditionally independent given $\bS$, the first term is zero, and the second term becomes $\Cov(S_j^*/c_j^*,S_k^*/c_k^*)=$ \linebreak $\Cov(\sum_{l\neq j}S_{jl},\sum_{l\neq k}S_{kl})/(c_j^*c_k^*)$, after removing the additive constants $s_0$. Applying the iterative formulae for a second time, the numerator is $\E\{\Cov(\sum_{l\neq j}S_{jl},\sum_{l\neq k}S_{kl}\mid W)\}+\Cov\{\E(\sum_{l\neq j}S_{jl}\mid W),\E(\sum_{l\neq k}S_{kl}\mid W)\}$. The only common term in the two sums is $S_{jk}$, therefore, the first term reduces to $\E\{\V(S_{jk}\mid W)\}=c_{jk}\E\{V(W)\}$. Calculating this last expected value and using the expression \eqref{eq:qvf}, we get $\E\{V(W)\}=c_0\V(W)$. The second term of the latter covariance reduces to $(\sum_{l\neq j}c_{jl})(\sum_{l\neq k}c_{kl})\V(W)$. Reassembling the original covariance we obtain, $$\Cov(Y_j,Y_k)=\frac{c_0c_{jk}+(\sum_{l\neq j}c_{jl})(\sum_{l\neq k}c_{kl})}{c_j^*c_k^*}V(W).$$ Finally, since $\V(Y_j)=\V(Y_k)=\V(W)$ we obtain the result. 
\item[(c)] If $c_{jk}=0$ then $\V(S_{jk}\mid W)=0$, therefore the latent link $S_{jk}=0$ w.p.1, and since $S_{jk}$ is the only direct link between $Y_j$ and $Y_k$, these latter two become conditionally independent. $\diamond$
\end{enumerate}
\end{proof2}

The exponential family graphical model \eqref{eq:expf} is defined by links denoted as $S_{ij}$, whereas in the normal graphical model \eqref{eq:normal} the links are denoted as $Z_{ij}$. The change of notation is because we would have to define $S_{ij}=c_{ij}Z_{ij}$ to make both constructions equivalent. We decided to use $Z_{ij}$ in the normal case because, in this case, equations \eqref{eq:normal} have the traditional form of the normal-normal conjugate model in a Bayesian analysis. 

There are six particular cases of members of the exponential family whose variance function is quadratic, as in \eqref{eq:qvf}. These are characterised by their functions $b(\cdot,\cdot)$, $h(\cdot,\cdot)$, $\theta(\cdot)$, $M(\cdot)$ and $J_\theta(\cdot)$. So the particular non Gaussian graphical models for the nodes $\{Y_j\}$, $j=1,\ldots,p$ that are derived from the three level hierarchical model \eqref{eq:expf} are:
\begin{enumerate}
\item[(i)] \textbf{Normal}: A normal graphical model uses links whose conditional distribution is again normal. The variance function is $V(\mu)=1$ and the parameter space is $\MC=\RB$. In particular, $b(s_{jk},c_{jk})=\left(2\pi c_{jk}\right)^{-1/2}\exp\left\{-{s_{jk}^2}/{(2c_{jk})}\right\}$, $h(s_0,c_0)=\left(2\pi/c_0\right)^{-1/2}\exp\left\{-s_0^2/(2c_0)\right\}$, $\theta(w)=w$, $M(\theta(w))=w^2/2$ and $J_\theta(w)=1$. In summary,
\begin{equation}
\nonumber
W\sim\no(s_0/c_0,c_0),\quad S_{jk}\mid W \sim \no(c_{jk}w,1/c_{jk})\quad\mbox{and}\quad Y_{j}\mid\bS \sim \no(s_j^*/c_j^*, c_j^*). 
\end{equation} 
Note that in \eqref{eq:normal} $m_0=s_0/c_0$ and $Z_{jk}=S_{jk}/c_{jk}$ to make both constructions equivalent. 
\item[(ii)] \textbf{Gamma}: A gamma graphical model uses links whose conditional distribution is Poisson. The variance function is $V(\mu)=\mu$ and the parameter space is $\MC=\RB^+$. In particular, $b(s_{jk},c_{jk})=c_{jk}^{s_{jk}}/s_{jk}!$, $h(s_0,c_0)=c_0^{s_0}/\Gamma(s_0)$, $\theta(w)=\log(w)$, $M(\theta(w))=w$ and $J_\theta(w)={1}/{w}$. In summary,
\begin{equation}
\nonumber
W\sim\ga(s_0,c_0),\quad S_{jk}\mid W \sim \po(c_{jk} w)\quad\mbox{and}\quad Y_{j}\mid\bS \sim \ga(s_j^*, c_j^*). 
\end{equation}
\item[(iii)] \textbf{Inverse Gamma}: An inverse gamma graphical model uses links whose conditional distribution is gamma. The variance function is $V(\mu)=\mu^2$ and the parameter space is $\MC=\RB^+$. In particular, $b(s_{jk},c_{jk})=s_{jk}^{c_{jk}-1}/\Gamma(c_{jk})$, $h(s_0,c_0)=s_0^{c_0}/\Gamma(c_0)$, $\theta(w)=-{1}/{w}$, $M(\theta(w))=\log(w)$ and $J_\theta(w)={1}/{w}$. In summary,
\begin{equation}
\nonumber
W\sim\iga(c_0+1,s_0),\quad S_{jk} \mid W \sim \ga(c_{jk},1/w)\quad\mbox{and} \quad 
Y_{j} \mid \bS \sim \iga(c_j^*+1, s_j^*).
\end{equation}
\item[(iv)] \textbf{Beta}: A beta graphical model uses links whose conditional distribution is binomial. The variance function is $V(\mu)=\mu-\mu^2$ and the parameter space is $\MC=(0,1)$. In particular, $b(s_{jk},c_{jk})={c_{jk}\choose s_{jk}}$, $h(s_0,c_0)=\Gamma(c_0)/\{\Gamma(s_0)\Gamma(c_0-s_0)\}$, $\theta(w)=\log\left\{{w}/{(1-w)}\right\}$, $M(\theta(w))=-\log(1-w)$ and $J_\theta(w)={1}/{\{w(1-w)\}}$. In summary,
\begin{equation}
\nonumber
W\sim\be(s_0,c_0-s_0),\quad S_{jk} \mid W \sim \bin(c_{jk},w) \quad \mbox{and} \quad 
Y_j \mid \bS \sim \be(s_j^*,c_j^*-s_j^*). 
\end{equation}
\item[(v)] \textbf{Inverse Beta}: An inverse beta graphical model uses links whose conditional distribution is negative binomial. The variance function is $V(\mu)=\mu+\mu^2$ and the parameter space is $\MC=\RB^+$. In particular, $b(s_{jk},c_{jk})={c_{jk}+s_{jk}-1\choose s_{jk}}$, $h(s_0,c_0)=\Gamma(s_0+c_0+1)/\{\Gamma(s_0)\Gamma(c_0+1)\}$, $\theta(w)=\log\left\{{w}/{(w+1)}\right\}$, $M(\theta(w))=\log(w+1)$ and $J_\theta(w)={1}/{\{w(w+1)\}}$. In summary,
\begin{equation}
\nonumber
W \sim \ibe(s_0,c_0+1),\quad S_{jk}\mid W \sim \nbi(c_{jk},1/(w+1))\quad\mbox{and}\quad Y_j\mid\bS \sim \ibe(s_j^*,c_j^*+1). 
\end{equation}
\item[(vi)] \textbf{Generalised Scaled Student}: A generalised scaled student graphical model uses links whose conditional distribution is generalised hyperbolic secant. The variance function is $V(\mu)=1+\mu^2$ and the parameter space is $\MC=\RB$. In particular, $b(s_{jk},c_{jk})=\{{2^{c_{jk}-2}}/{\Gamma(c_{jk})}\} \prod_{l=0}^\infty\left\{1+{s_{jk}^2}/{(c_{jk}+2l)^2}\right\}^{-1}$, $h(s_0,c_0)$ is a normalising constant, $\theta(w)=\tan^{-1}(w)$, $M(\theta(w))=({1}/{2})\log\left(1+w^2\right)$ and $J_\theta(w)=\left(1+w^2\right)^{-1}$. In summary,
\begin{equation}
\nonumber
W\sim \gsst(s_0/c_0,c_0),\quad S_{jk}\mid W \sim \ghs(c_{jk} w,1/c_{jk})\quad\mbox{and}\quad Y_j\mid\bS \sim \gsst(s_j^*/c_j^*, c_j^*).  
\end{equation}
\end{enumerate}

We therefore have a set of different graphical models with support in the real numbers (i) and (vi) cases, positive numbers (ii), (iii) and (v) cases, and with bounded support in the unit interval (iv) case. 

\section{Bayesian inference}
\label{sec:bayes}


Let $\bY_i'=(Y_{i1},\ldots,Y_{ip})$, $i=1,\ldots,n$ be a sample of $n$ replicates from the network \eqref{eq:expf} with $p$ nodes 
Since the law of the graph is given by a three-level hierarchical model, we assume that together with $\bY_i$ we have also observed latent variables $\bS_i$ and $W_i$, as in a data augmentation technique \citep{tanner:91}. Therefore the extended likelihood for $\bfeta_i=(s_{0},c_{0},\bc_i)$ in terms of $(\bY_i,\bS_i,W_i)$, for $i=1\ldots,n$ is of the form 
\begin{equation}
\label{eq:lik}
f(\by,\bs,\bw\mid\bfeta)=\prod_{i=1}^n\left[\left\{\prod_{j=1}^p f(y_{ij}\mid s_{ij}^*,c_{ij}^*)\right\}\left\{\prod_{j<k}^p f(s_{ijk}\mid \theta(w_i),c_{ijk})\right\}f(w_i\mid s_{0},c_{0})\right].
\end{equation}

Our main parameters of interest are $c_{jk}$, which determine whether a link is present ($c_{ijk}>0$) or not ($c_{ijk}=0$). To avoid further complexities in the model, we consider that $c_{jk}$ are positive parameters and assign them a prior distribution with positive support. In the end, we will report the estimated value of $c_{jk}$, which will be interpreted as an intensity of the edge that links nodes $j$ and $k$. 

Specifically, the prior distributions for these parameters are of the form: $c_{jk}\sim\ga(a_0,b_0)$ for the normal, gamma, inverse gamma and generalised scaled student cases; and $c_{jk}\sim\po(d_0)$ for the beta and inverse beta cases. 

The other two parameters $s_0$ and $c_0$ control the marginal distribution of each of the nodes $Y_{ij}$ for $j=1,\ldots,p$. In particular, $\E(Y_{ij})=s_{0}/c_{0}$. Typical graphical Gaussian models assume a zero mean \citep{gan&al:19}, which is not always possible in our exponential family models. In particular, we take $c_0\sim\ga(a_c,b_c)$ for all cases, and $s_0\sim\no(\mu_s,\tau_s)$ for the normal and generalised scaled student cases; $s_0\sim\ga(a_s,b_s)$ for the gamma, inverse gamma and inverse beta cases; and $s_0\mid c_0\sim\un(0,c_0)$ for the beta case. 

Note that we do not require further constraints so that the variance-covariance matrix of $\bY_i$ is positive defined, as is common in Gaussian graphical models. 

Posterior distribution of model parameters will be characterised through the full conditional distributions, which are given as follows. 
\begin{enumerate}
\item[(I)] The full conditional distribution of $s_{0}$ becomes
$$f(s_{0}\mid\rest)\propto\left\{\prod_{i=1}^n\prod_{j=1}^p h(s_{ij}^*,c_{ij}^*)\right\}\left\{h(s_{0},c_{0})\right\}^n\exp\left[s_{0}\sum_{i=1}^n\left\{\sum_{j=1}^p\theta(y_{ij})+\theta(w_i)\right\}\right]f(s_{0}),$$
where $f(s_{0})$ is the prior distribution according to the choice of model (i)--(vi). 
\item[(II)] The full conditional distribution of $c_{0}$ becomes
$$f(c_{0}\mid\rest)\propto\left\{\prod_{i=1}^n\prod_{j=1}^p h(s_{ij}^*,c_{ij}^*)\right\}\left\{h(s_{0},c_{0})\right\}^n\exp\left[-c_{0}\sum_{i=1}^n\left\{\sum_{j=1}^p M(\theta(y_{ij}))+M(\theta(w_i))\right\}\right]$$
$$\hspace{-6.5cm}\times\,\ga(c_{0}\mid a_c,b_c).$$
\item[(III)] The full conditional distribution for $c_{jk}$, $j<k=1,\ldots,p$,  becomes
$$\hspace{-7cm}f(c_{jk}\mid\rest)\propto\left\{\prod_{i=1}^n h(s_{ij}^*,c_{ij}^*)h(s_{ik}^*,c_{ik}^*)b(s_{ijk},c_{ijk})\right\}$$
$$\hspace{1.2cm}\times\,\exp\left[-\sum_{i=1}^n c_{ijk}\left\{M(\theta(y_{ij}))+M(\theta(y_{ik}))+M(\theta(w_i))\right\}\right]f(c_{jk}\mid a_0,b_0).$$
\end{enumerate}
In all cases, $s_{ij}^*=s_0+\sum_{k\ne j}s_{ijk}$ and $c_{ij}^*=c_0+\sum_{k\ne j}c_{ijk}$. 

Recall that we have assumed that we have observed the latent variables $\bS_i$ and $W_i$, so to be able to make inference we have to sample from the posterior conditional distributions of the latent variables, which are given as follows. 
\begin{enumerate}
\item[(IV)] The full conditional distribution of $S_{ijk}$, $i=1,\ldots,n$, $j<k=1,\ldots,p$ becomes
$$f(s_{ijk}\mid\rest)\propto h(s_{ij}^*,c_{ij}^*)h(s_{ik}^*,c_{ik}^*)b(s_{ijk},c_{ijk})\exp\left[s_{ijk}\left\{\theta(y_{ij})+\theta(y_{ik})+\theta(w_i)\right\}\right].$$
\item[(V)] The full conditional distribution of $W_i$, $i=1,\ldots,n$ becomes
$$f(w_i\mid\rest)= h(s_i^*,c_i^*)\exp\left\{\theta(w_i)s_i^*-c_i^* M(\theta(w_i))\right\}|J_\theta(w_i)|,$$
where $s_i^*=s_{0}+\sum_{j<k}s_{ijk}$ and $c_i^*=c_{0}+\sum_{j<k}c_{ijk}$. Due to conjugacy, this conditional distribution has the same form as $f(w_i\mid s_0,c_0)$, first equation in \eqref{eq:expf}, but with updated parameters. 
\end{enumerate}

Posterior inference will therefore require the implementation of a Gibbs sampler \citep{smith&roberts:93} with conditional distributions (I)--(V). Whenever is not possible to sample directly from the posterior conditional, we implement a random walk Metropolis-Hastings step \citep{tierney:94}. Alternatively, since apart from the generalised scaled student case, the other exponential family members are all of standard form, we can perform inference using JAGS \citep{plummer:23}. JAGS code for models (i), (ii) and (iii) can be found in the Appendix.

\section{Numerical illustrations}
\label{sec:app}

\subsection{Simulated data}

To illustrate the performance of our model and the way we are going to report inferences, we first consider a control scenario where we sample data $Y_{ij}$ from the inverse gamma graphical model (iii) with parameters $s_0=6$, $c_0=10$ and $c_{jk}\in\{e^{2},e^{-2}\}$ according to whether regions $j$ and $k$ are or are not neighbours, respectively. To be specific, we consider a hypothetical region with $p=5$ areas as the one depicted in Figure \ref{fig7:region}. Here, only the pairs of regions $(1,3)$, $(1,5)$ and $(3,4)$ are not neighbours, the rest are. The chosen sample size is $n=500$. 

This setting defines a network with $p=5$ nodes and edges with intensity $c_{jk}$. To show this in a graph, we associate the value of $c_{jk}$ with a gray intensity by normalising each $c_{jk}$ dividing by $\max_{j,k}c_{jk}$, multiplying by 100 and rounding to two decimal places. In doing this, our controlled scenario corresponds to the network shown in the first panel of Figure \ref{fig:simnet}. We can clearly see that the only edges not shown are those for which the value of $c_{jk}$ is very close to zero. 

We fitted an inverse gamma graphical model with prior specifications given by: $a_0=b_0=1$, $a_c=b_c=0.01$ and $a_s=b_s=0.01$. Implementation was done in JAGS with 105,000 iterations, 15,000 as burn-in and a thinning of 40. We also ran two parallel chains to assess convergence. The running time was 3 hours in a a core i7 at 3.40 GHz and 32 GB of RAM. 

The posterior estimated values for $c_{jk}=0.14$ are: $c_{13}\in(0.11,1.84)$, $c_{15}\in(0.07,2.48)$ and $c_{34}\in(0.13,2.57)$ with 95\% of probability; and posterior estimates for $c_{jk}=7.39$ are: $c_{12}\in(3.58,9.57)$, $c_{14}\in(3.49,9.16)$, $c_{23}\in(2.92,8.58)$, $c_{24}\in(2.88,9.46)$, $c_{25}(3.43,9.76)$, $c_{35}\in(2.63,7.41)$ and $c_{45}\in(2.62,9.07)$ with 95\% of probability. This confirms that, although the intervals are a little wide, the inferential procedure is capturing the true values. Finally, we also produce the inferred network using the posterior mean values as point estimates for $c_{jk}$. This graph is shown in the right panel on the same Figure \ref{fig:simnet}. We can see that the estimated network is very similar to the true one. 

\subsection{Sunspots data}

The number of sunspots on a given day has been studied since the year 1610. The monthly mean sunspots numbers have been reported in \cite{andrews&herzberg:85} since 1749 until 1983. We consider the twelve months as nodes in a graph, i.e. $p=12$ and the years as replicates of the counts for $i=1,\ldots,235$. Our objective is to identify the conditionally dependence strength among months. 

Data are shown as monthly box plots in Figure \ref{fig:sunbox}. Monthly distributions are stationary and right-skewed, which suggests that the inverse gamma or gamma models are good options to describe their behaviour. For both models, prior distributions were defined as in the simulation study. Inference was done in JAGS with two chains defined with a total of 70,000 iterations, 10,000 as burn-in and a thinning of 40. The running time was 11.5 hours for the inverse gamma model and 6.5 hours for the gamma model, almost half the time. 

To compare the fit of the models to the data, we computed the mean square error defined as $MSE=\frac{1}{np}\sum_{i=1}^n\sum_{j=1}^p\left(Y_{ij}^F-Y_{ij}\right)^2$, where $Y_{ij}^F$ is the predicted value associated to observation $Y_{ij}$. The posterior expected values of $MSE$, in the inverse gamma model, is $10.70$ and, in the gamma model, is $2.44$. This suggests that the latter provides the best fit. 

Instead of reporting the estimated conditional dependence parameters $c_{jk}$, we normalize them by their maximum estimated value and produce a network. The inferred networks with both models are included in Figure \ref{fig:sunnet}. Although the best fit is obtained with the gamma model, both networks are similar. We see that there are some pairs of months whose conditional dependence is stronger, say May-June, July-August, September-October, and November-December. In a second degree of conditionally dependence we have the first three months of the year, January-February-March. What is surprising is that December-January are almost conditionally independent, in some way the first month of the year establishes a new beginning in the number of sunspots, which is conditionally independent of the number of spots at end of the previous year. Considering nonadjacent months, the only edges that can not be disregarded are those that link April-July, April-June and June-August. 

\subsection{Glucose data}

For early detection of diabetis mellitus, an oral glucose tolerance test is performed. The test is subject to considerable variation in pregnant and non-pregnant women. A set of women underwent annual glucose tolerance tests for a period of three years. Measurements were taken at fasting and one hour after an intake of 100 grams of glucose. We have a total of $p=6$ measurements, where the first three were taken at fasting and the last three after glucose intake. The data \citep{andrews&herzberg:85} is split into two parts. The first part consists of $n_1=53$ non-pregnant women and the second part consists of $n_2=52$ pregnant women. The objective is to identify conditionally independent differences among the six nodes in pregnant and non-pregnant women. 

To make the distribution of the six nodes stationary, we standardized each of the six variables and added the value of 3 to move the numbers away from zero. Since standardization is a linear transformation, this does not affect the correlations. The box plots of the six variables are included in Figure \ref{fig:glubox}, for non-pregnant women (left panel) and pregnant women (right panel). In both datasets the box plots look very symmetric; therefore, apart from the inverse gamma and the gamma models, we will also consider the normal model. 

Again, the prior distributions were the same as in the simulation study. Inference was implemented in JAGS with the chains specified as for the sunspots data. The running times were around 27 minutes for the inverse gamma model, 20 minutes for the gamma model, and 6 minutes for the normal model. 

The posterior expected values of $MSE$s for each model in both datasets are given in Table \ref{tab:glumse}. For both datasets, the worst fit is obtained by the inverse gamma model, followed by the gamma model in second place, and the best fit is obtained by the normal model. This latter is not only the best, but also the fasted to fit. 

Our posterior summary, for conditionally dependent parameters $c_{jk}$, is presented in Figure \ref{fig:glunet}, for the gamma model (left column), normal model (right column), non-pregnant women dataset (top row) and pregnant women dataset (bottom row). Again, inferences obtained with both models are similar. 

For the non-pregnant women, conditionally dependencies appear everywhere measurements 1-4, 2-5 and 3-6, which were made before and after glucose intake, are dependent, as well as measurements in consecutive years after glucose intake 4-5 and 4-6. 

For pregnant women, the highest dependence occurs between measurements 4-5, both of which were made after glucose intake, the rest of the dependencies are a lot lower. This confirms that variability and dependence in glucose measurement tests is different in non-pregnant and pregnant women. 

\section{Discussion}
\label{sec:disc}

We have successfully proposed new likelihoods for the probabilistic study of graphs. Our construction relies on the inclusion of pairwise latent variables, each with a specific edge parameter. Although we increment the number of random quantities in the model, the advantage is that posterior simulation is very efficient, avoiding the inclusion of positive-definiteness constraints. 

Exponential family members with quadratic variance are well-known distributions, apart from one. This makes the implementation straightforward in simple code like JAGS, which can be run through R-package \citep{r:24}. Depending on the size of the nodes $p$ and the sample size $n$, running times can increase. 

Future directions of study are the inclusion of covariates or the analysis of other sampling schemes such as time series data. 

\section*{Acknowledgements}

This work was supported by \textit{Asociaci\'on Mexicana de Cultura, A.C.}

\bibliographystyle{natbib}

\section*{Appendix}

JAGS code for graphical models (i), (ii) and (iii). 
\footnotesize{ \begin{verbatim}
# ---------------------------------------------------------------
# Bugs code graphical model with normal marginals
model{
#Likelihood
  for (i in 1:n) {
    w[i] ~ dnorm(s0/c0,c0)
    s[i,1,1] <- 0
    for (j in 1:(p-1)){
    s[i,j+1,j+1] <- 0
    for (k in (j+1):p){
      s[i,j,k] ~ dnorm(c[j,k]*w[i],1/c[j,k])
      s[i,k,j] <- s[i,j,k]
    }
    }
    for (j in 1:p){
      ay[i,j] <- s0+sum(s[i,j,1:p])
      by[i,j] <- c0+sum(c[j,1:p])
      y[i,j] ~ dnorm(ay[i,j]/by[i,j],by[i,j])
    }
  }
#Prior
  c0 ~ dgamma(0.01,0.01)
  s0 ~ dnorm(0,0.01)
  c[1,1] <- 0
  for (j in 1:(p-1)){
    c[j+1,j+1] <- 0
    for (k in (j+1):p){
      c[j,k] ~ dgamma(a0,b0)
      c[k,j] <- c[j,k]
    }
  }
#Prediction
  for (i in 1:n){
  for (j in 1:p){
    yf[i,j] ~ dnorm(ay[i,j]/by[i,j],by[i,j])
    dif[i,j] <- pow(y[i,j]-yf[i,j],2)
  }
  }
  mse <- mean(dif[,])
}
\end{verbatim}}

\footnotesize{ \begin{verbatim}
# ---------------------------------------------------------------
# Bugs code graphical model with gamma marginals
model{
#Likelihood
  for (i in 1:n) {
    w[i] ~ dgamma(s0,c0)
    s[i,1,1] <- 0
    for (j in 1:(p-1)){
    s[i,j+1,j+1] <- 0
    for (k in (j+1):p){
      s[i,j,k] ~ dpois(c[j,k]*w[i])
      s[i,k,j] <- s[i,j,k]
    }
    }
    for (j in 1:p){
      ay[i,j] <- s0+sum(s[i,j,1:p])
      by[i,j] <- c0+sum(c[j,1:p])
      y[i,j] ~ dgamma(ay[i,j],by[i,j])
    }
  }
#Prior
  c0 ~ dgamma(0.01,0.01)
  s0 ~ dgamma(0.01,0.01)
  c[1,1] <- 0
  for (j in 1:(p-1)){
    c[j+1,j+1] <- 0
    for (k in (j+1):p){
      c[j,k] ~ dgamma(a0,b0)
      c[k,j] <- c[j,k]
    }
  }
#Prediction
  for (i in 1:n){
  for (j in 1:p){
    yf[i,j] ~ dgamma(ay[i,j],by[i,j])
    dif[i,j] <- pow(y[i,j]-yf[i,j],2)
  }
  }
  mse <- mean(dif[,])
}
\end{verbatim}}

\footnotesize{ \begin{verbatim}
# ---------------------------------------------------------------
# Bugs code for graphical model with inverse gamma marginals
# Input data has to be yy[i,j] = 1 / y[i,j], where y[i,j] are actual observations
model{
#Likelihood
  for (i in 1:n) {
    ww[i] ~ dgamma(c0+1,s0)
    w[i] <- 1/ww[i]
    s[i,1,1] <- 0
    for (j in 1:(p-1)){
    s[i,j+1,j+1] <- 0
    for (k in (j+1):p){
      s[i,j,k] ~ dgamma(c[j,k],ww[i])
      s[i,k,j] <- s[i,j,k]
    }
    }
    for (j in 1:p){
      ay[i,j] <- c0+sum(c[j,1:p])+1
      by[i,j] <- s0+sum(s[i,j,1:p])
      yy[i,j] ~ dgamma(ay[i,j],by[i,j])
      y[i,j] <- 1/yy[i,j]
    }
  }
#Prior
  c0 ~ dgamma(0.01,0.01)
  s0 ~ dgamma(0.01,0.01)
  c[1,1] <- 0
  for (j in 1:(p-1)){
    c[j+1,j+1] <- 0
    for (k in (j+1):p){
      c[j,k] ~ dgamma(a0,b0)
      c[k,j] <- c[j,k]
    }
  }
#Prediction
  for (i in 1:n){
  for (j in 1:p){
    yyf[i,j] ~ dgamma(ay[i,j],by[i,j])
    yf[i,j] <- 1/yyf[i,j]
    dif[i,j] <- pow(y[i,j]-yf[i,j],2)
  }
  }
  mse <- mean(dif[,])
}
\end{verbatim}}


\begin{table}
\centering
\begin{tabular}{lccc} \hline\hline 
Dataset & IGamma & Gamma & Normal \\ \hline 
Non-pregnant & 1.68 & 0.97 & 0.82 \\
Pregnant & 1.61 & 1.09 & 0.98 \\
\hline\hline
\end{tabular}
\caption{Posterior expected $MSE$ for the two datasets and different models.} 
\label{tab:glumse}
\end{table}


\begin{figure}[ht]
\setlength{\unitlength}{1.3cm}
\begin{center}
\begin{picture}(5,5)
\put(2.55,4){$Y_1$}
\put(4.05,1){$Y_2$}
\put(1.03,1){$Y_3$}
\put(2.7,4.1){\circle{0.6}}
\put(4.2,1.1){\circle{0.6}}
\put(1.2,1.1){\circle{0.6}}
\put(1.35,1.35){\line(1,2){1.22}}
\put(1.5,1.1){\line(1,0){2.4}}
\put(2.8,3.8){\line(1,-2){1.22}}
\put(2.5,0.7){$Z_{2,3}$}
\put(3.7,2.3){$Z_{1,2}$}
\put(1.2,2.4){$Z_{1,3}$}
\end{picture}
\end{center}
\vspace{-1.2cm}
\caption{Graphical representation of a network of size $n=3$.}
\label{fig:net3} 
\end{figure}
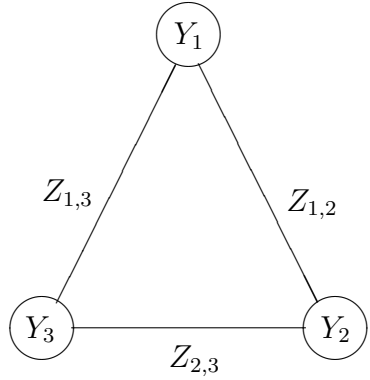

\begin{figure}[h]
\setlength{\unitlength}{0.8cm}
\begin{center}
\begin{picture}(5,7)
\put(2.5,6.6){$W$} 
\put(2.7,6.3){\vector(-2,-3){1.8}}
\put(2.7,6.3){\vector(0,-1){2.6}}
\put(2.7,6.3){\vector(2,-3){1.8}}
\put(0.4,3.1){$Z_{12}$}
\put(2.5,3.1){$Z_{13}$}
\put(4.4,3.1){$Z_{23}$}
\put(4.5,0.1){$Y_3$}
\put(2.5,0.1){$Y_2$}
\put(0.4,0.1){$Y_1$}
\put(0.6,2.8){\vector(0,-1){2.0}}
\put(0.6,2.8){\vector(1,-1){2.0}}
\put(2.7,2.8){\vector(-1,-1){2.0}}
\put(2.7,2.8){\vector(1,-1){2.0}}
\put(4.8,2.8){\vector(-1,-1){2.0}}
\put(4.8,2.8){\vector(0,-1){2.0}}
\end{picture}
\end{center}
\caption{Graphical representation of hierarchical model \eqref{eq:normal} for network in Figure \ref{fig:net3}.}
\label{fig:graphm} 
\end{figure}
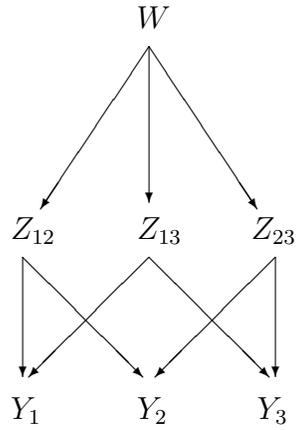

\begin{figure}[h]
\setlength{\unitlength}{0.8cm}
\begin{center}
\begin{picture}(6,4)
\put(0,0){\line(0,1){4}}
\put(6,0){\line(0,1){4}}
\put(0,0){\line(1,0){6}}
\put(0,2){\line(1,0){6}}
\put(0,4){\line(1,0){6}}
\put(2,4){\line(0,-1){2}}
\put(4,4){\line(0,-1){2}}
\put(3,2){\line(0,-1){2}}
\put(0.7,2.8){$Y_1$}
\put(2.7,2.8){$Y_2$}
\put(4.7,2.8){$Y_3$}
\put(1.2,0.8){$Y_4$}
\put(4.2,0.8){$Y_5$}
\end{picture}
\end{center}
\caption{Hypothetical region with five areas.}
\label{fig7:region} 
\end{figure}
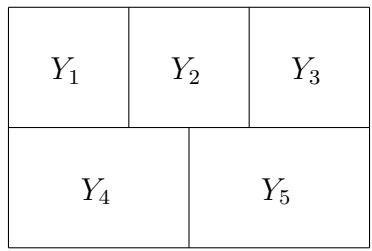

\begin{figure}
\centering
\includegraphics[scale=0.45]{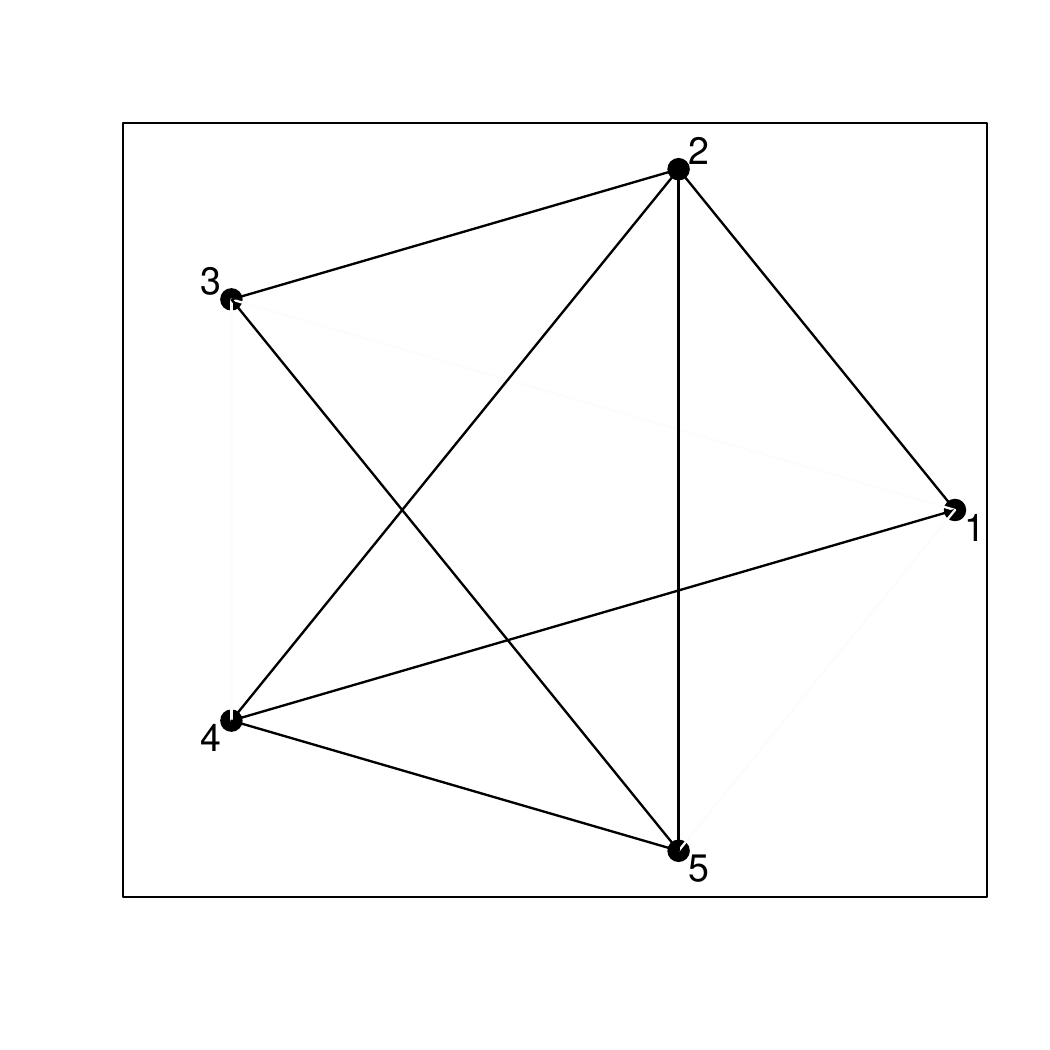}
\includegraphics[scale=0.45]{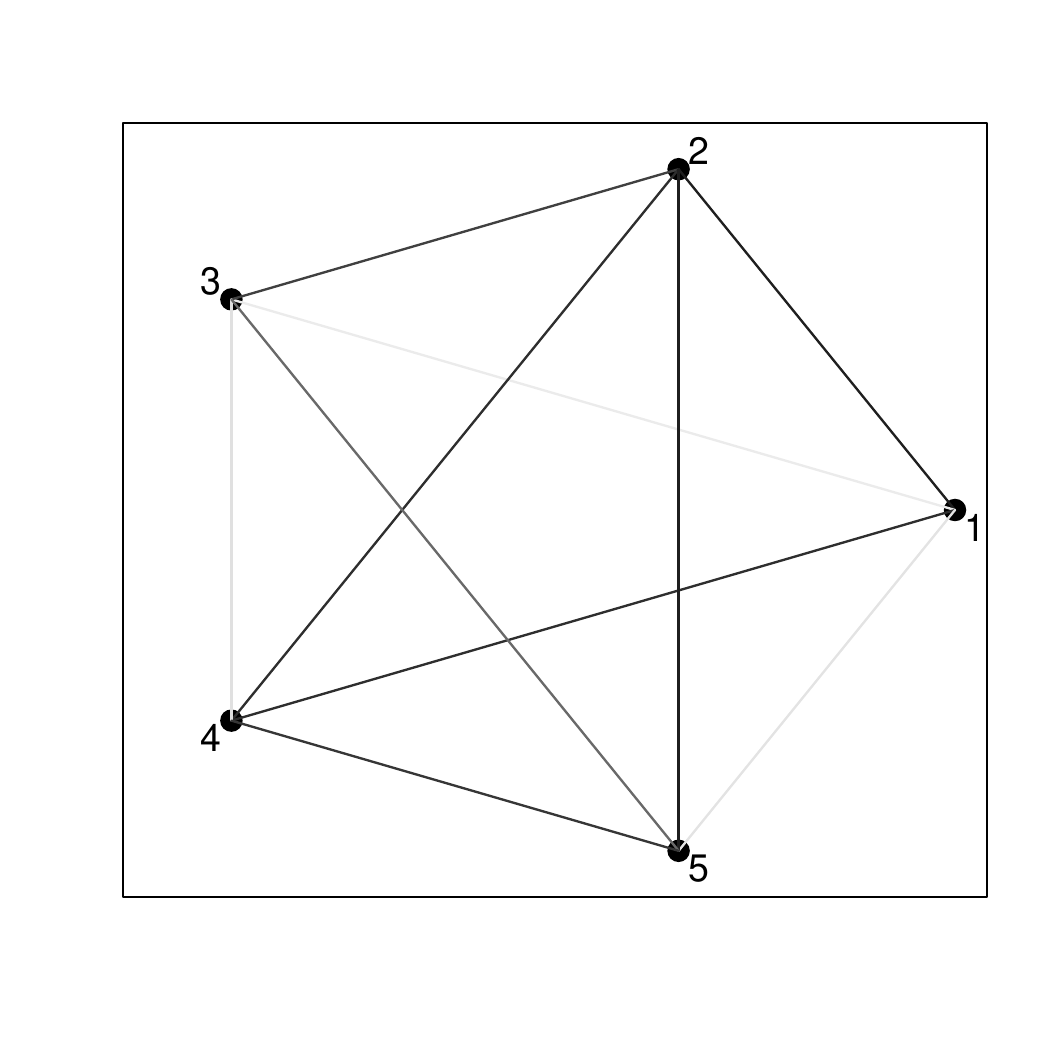}
\caption{Networks in simulation study. True (left), estimated with $n=500$ (right).}
\label{fig:simnet}
\end{figure}

\begin{figure}
\centering
\includegraphics[scale=0.9]{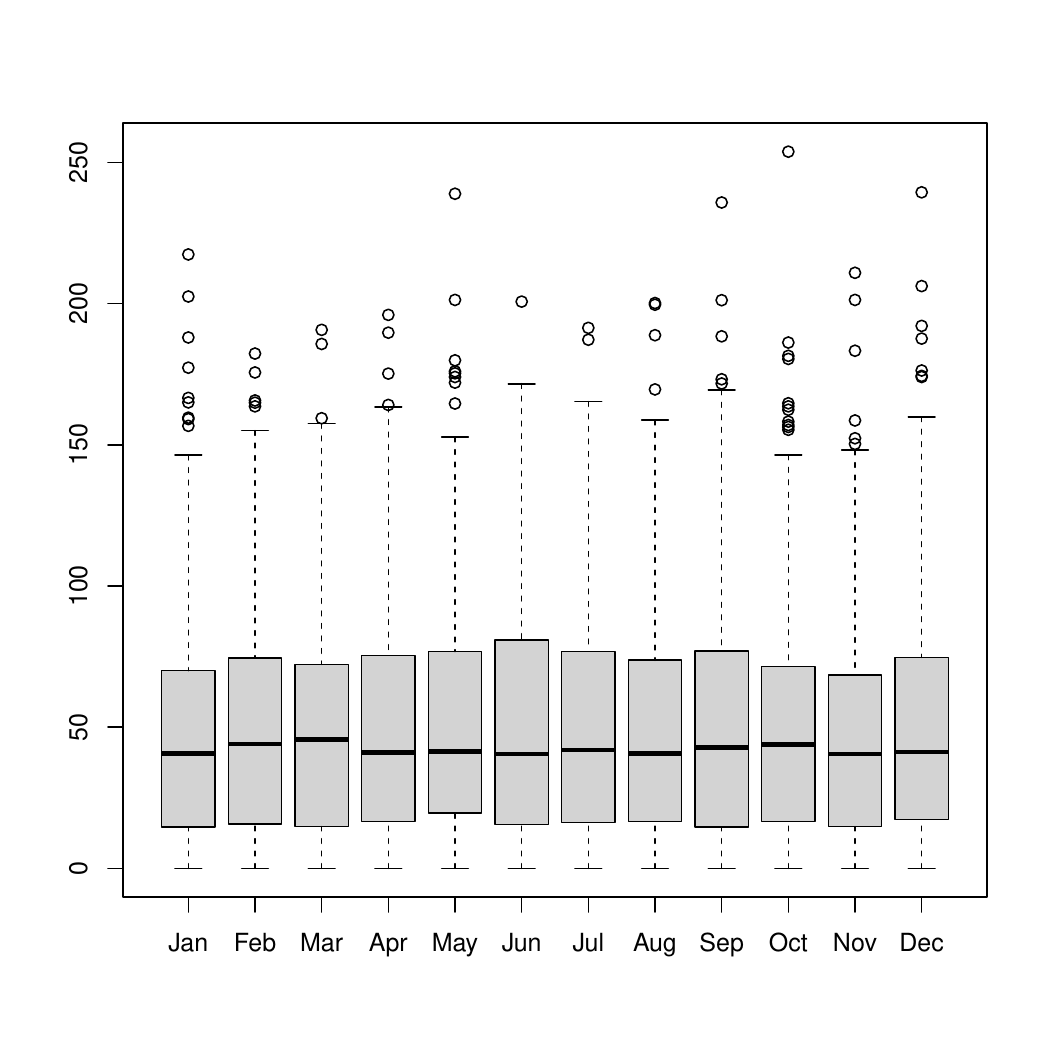}
\caption{Monthly box plots of the number of sunspots.}
\label{fig:sunbox}
\end{figure}

\begin{figure}
\centering
\includegraphics[scale=0.45]{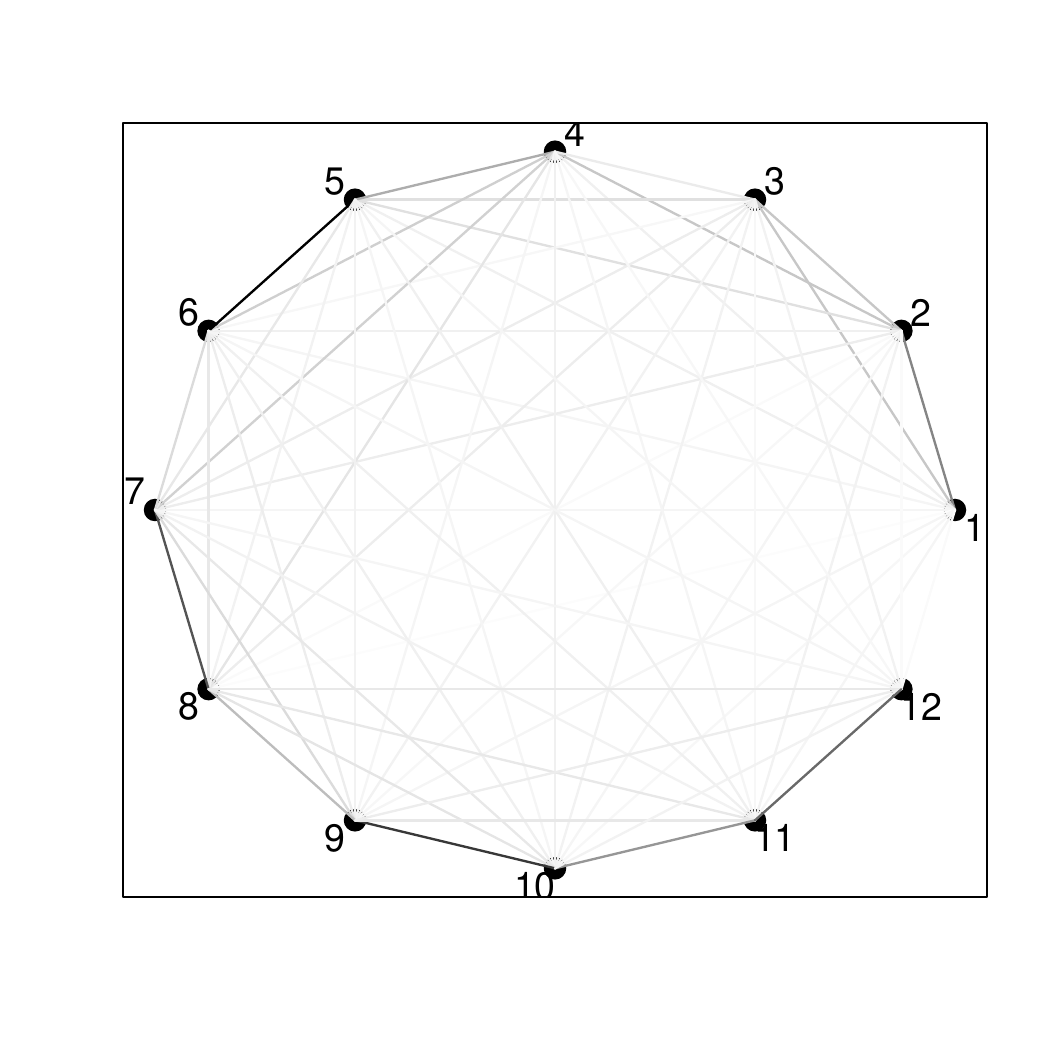}
\includegraphics[scale=0.45]{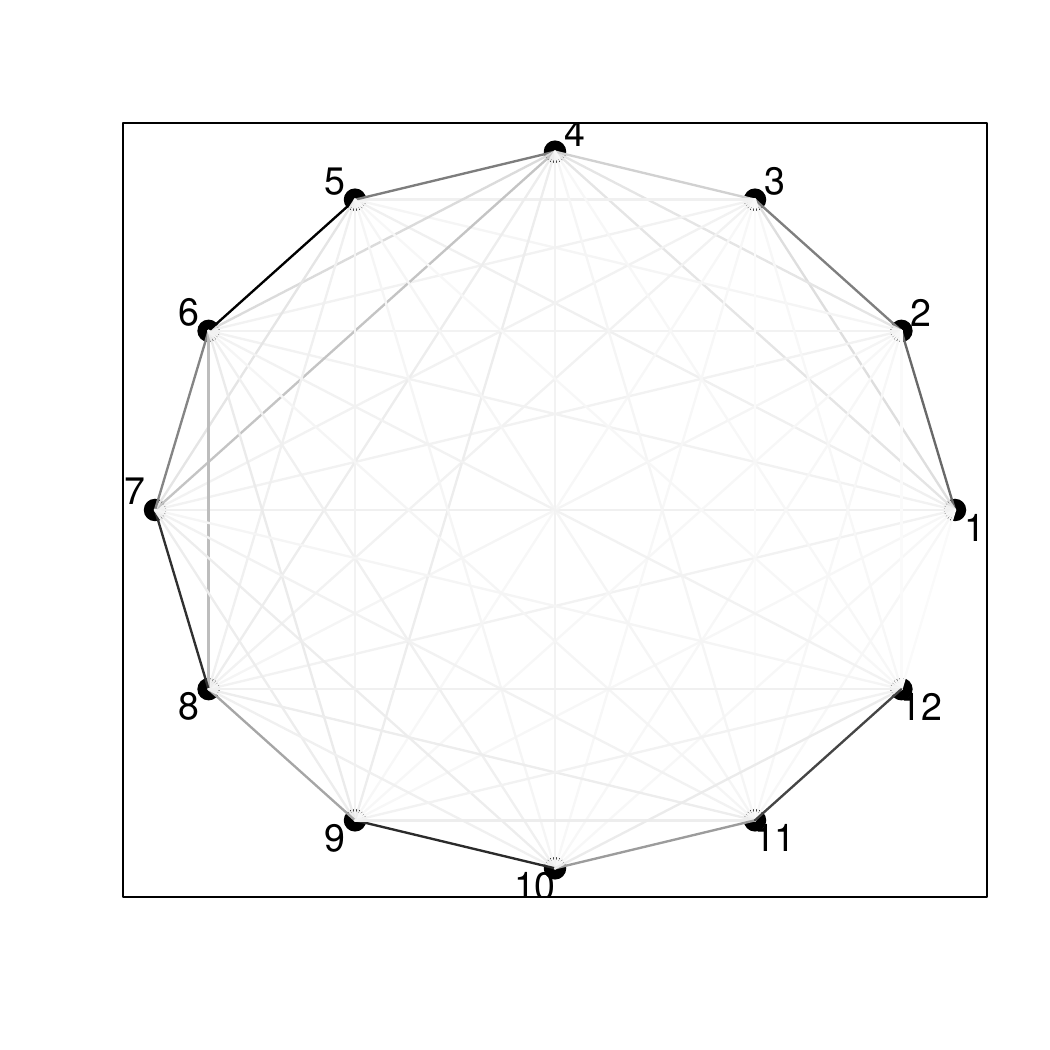}
\caption{Networks in sunspots dataset. Inverse gamma model (left), gamma model (right).}
\label{fig:sunnet}
\end{figure}

\begin{figure}
\centering
\includegraphics[scale=0.45]{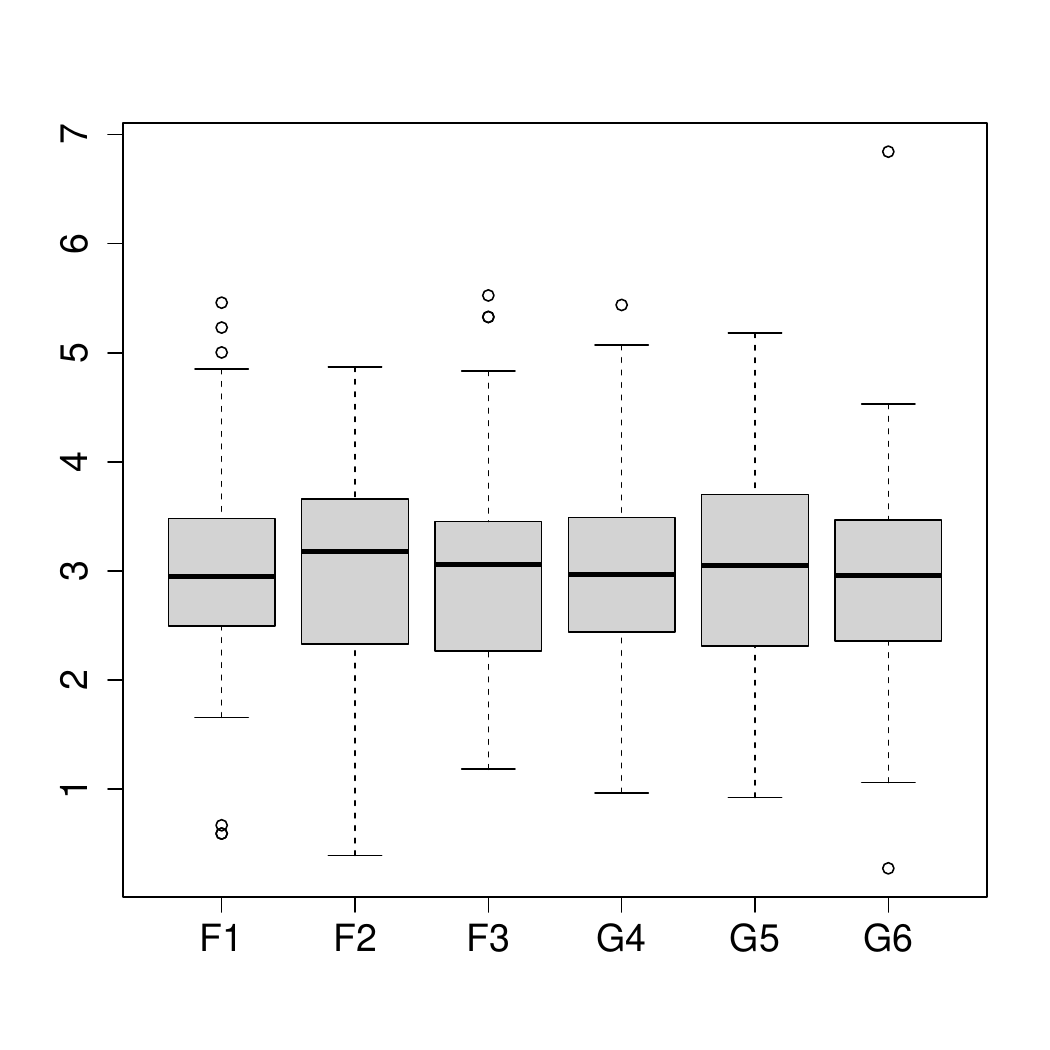}
\includegraphics[scale=0.45]{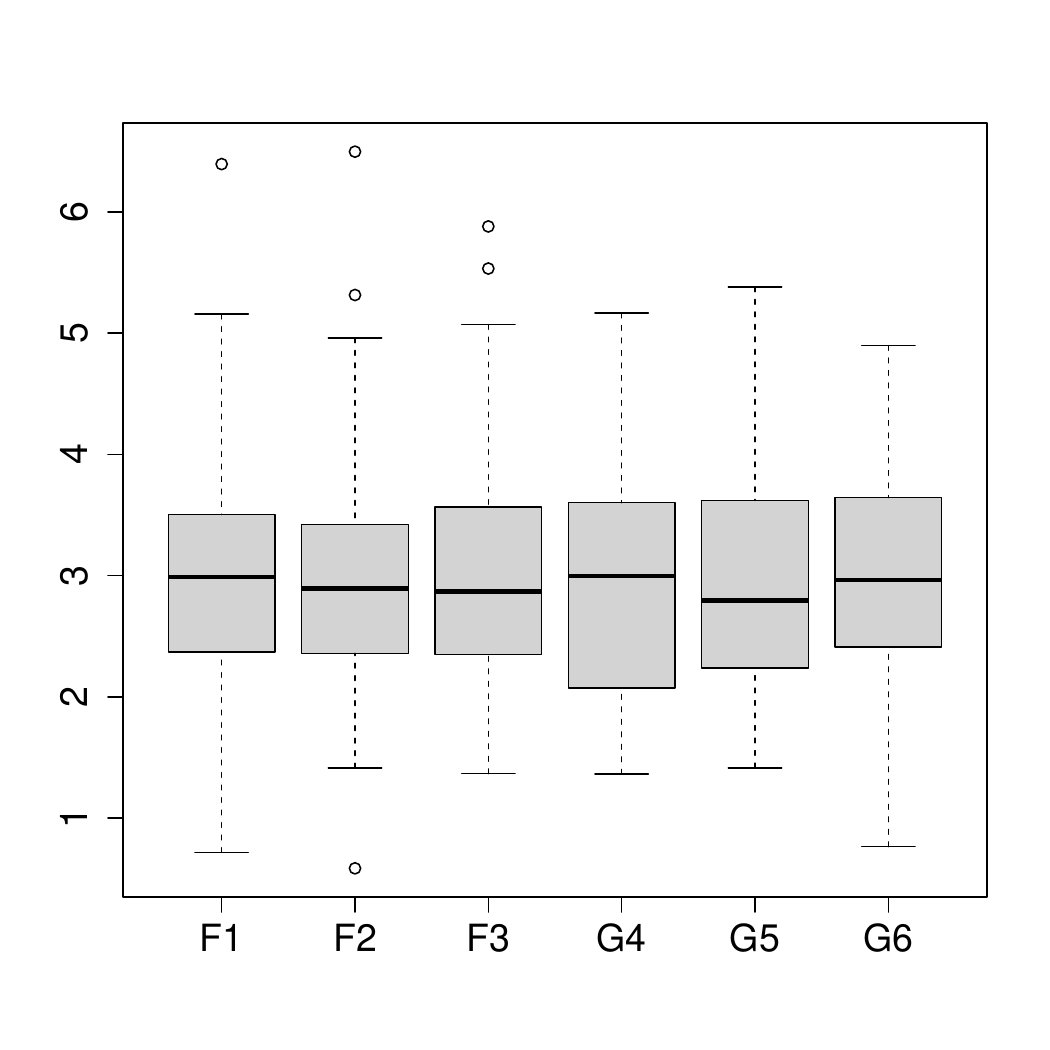}
\caption{Box plots of glucose blood test measurements at fasting (F) and after glucose intake for three years. Non-pregnant women (left), pregnant women (right).}
\label{fig:glubox}
\end{figure}

\begin{figure}
\centering
\includegraphics[scale=0.45]{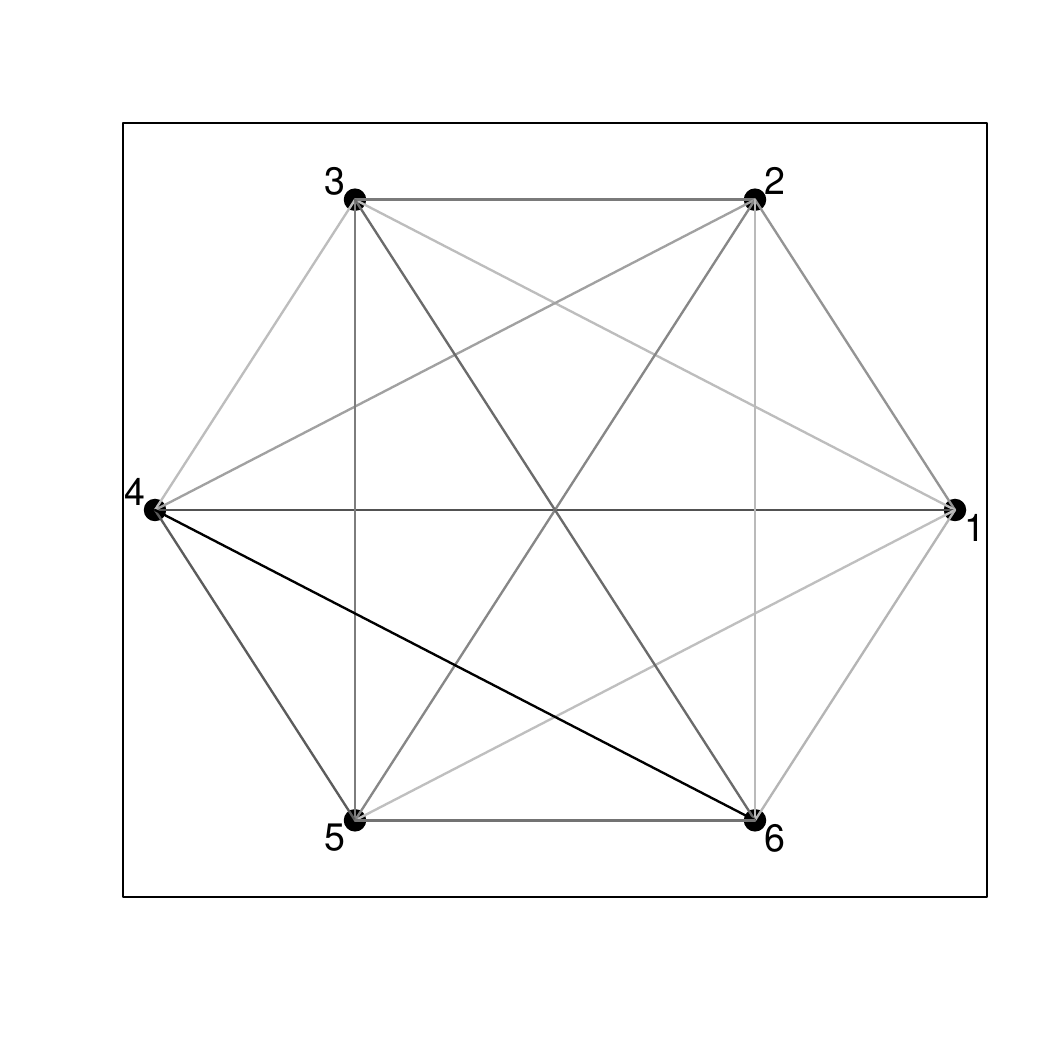}
\includegraphics[scale=0.45]{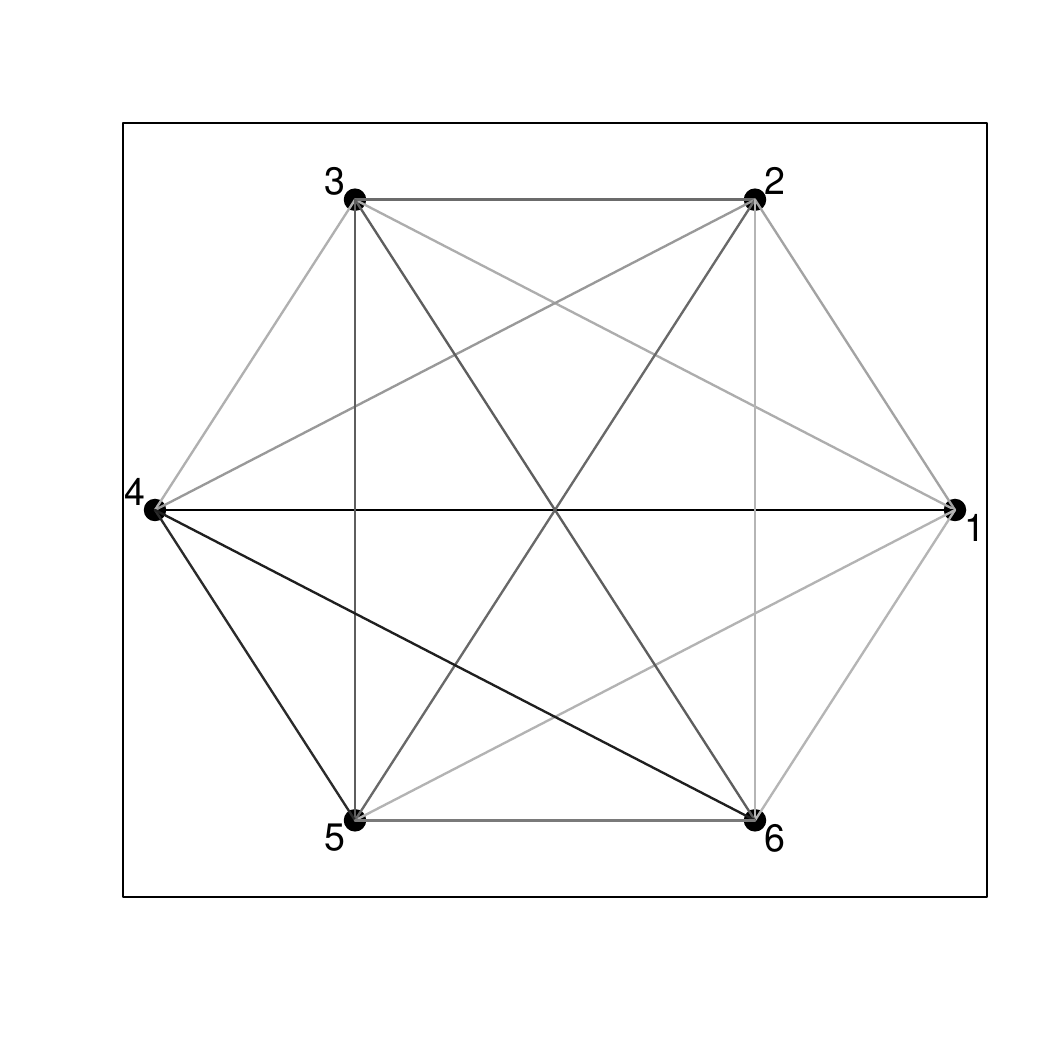}
\includegraphics[scale=0.45]{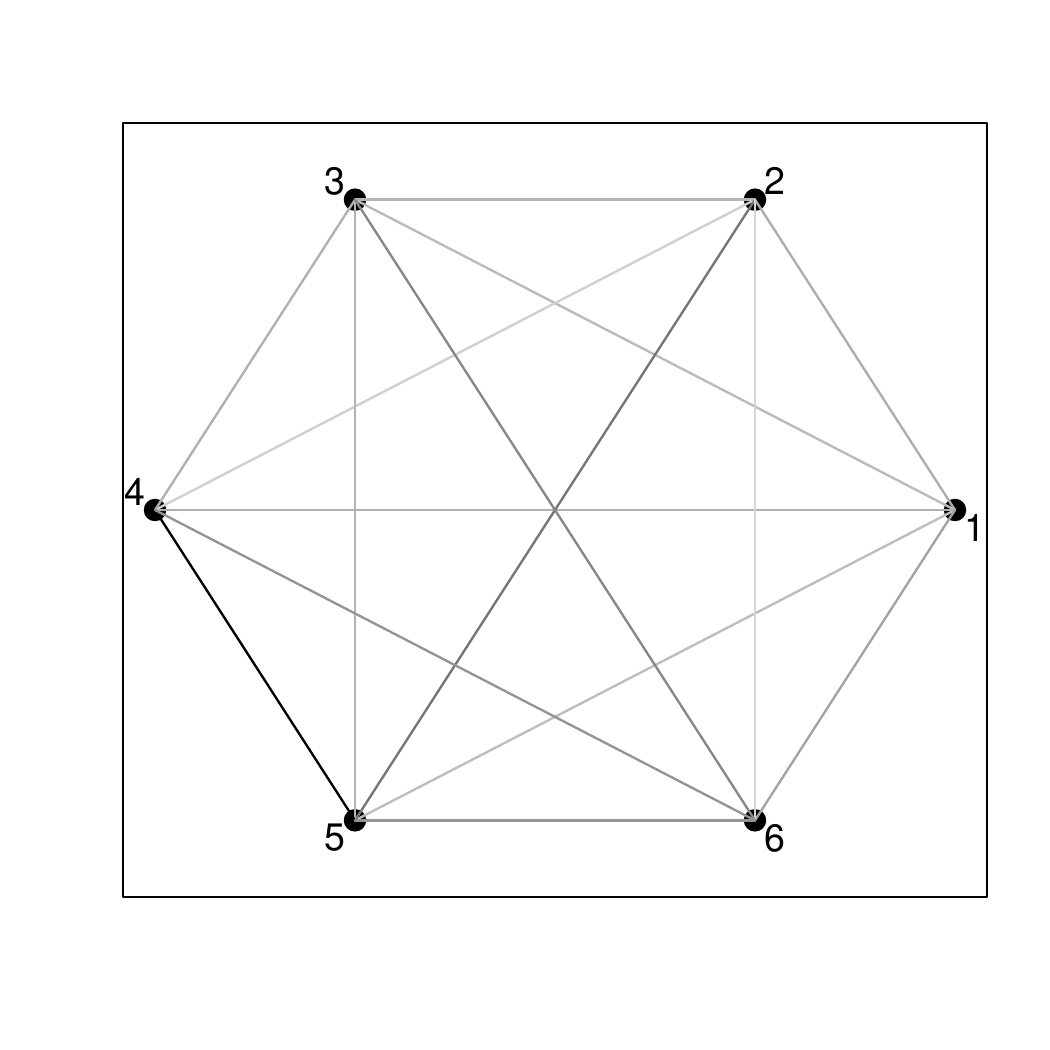}
\includegraphics[scale=0.45]{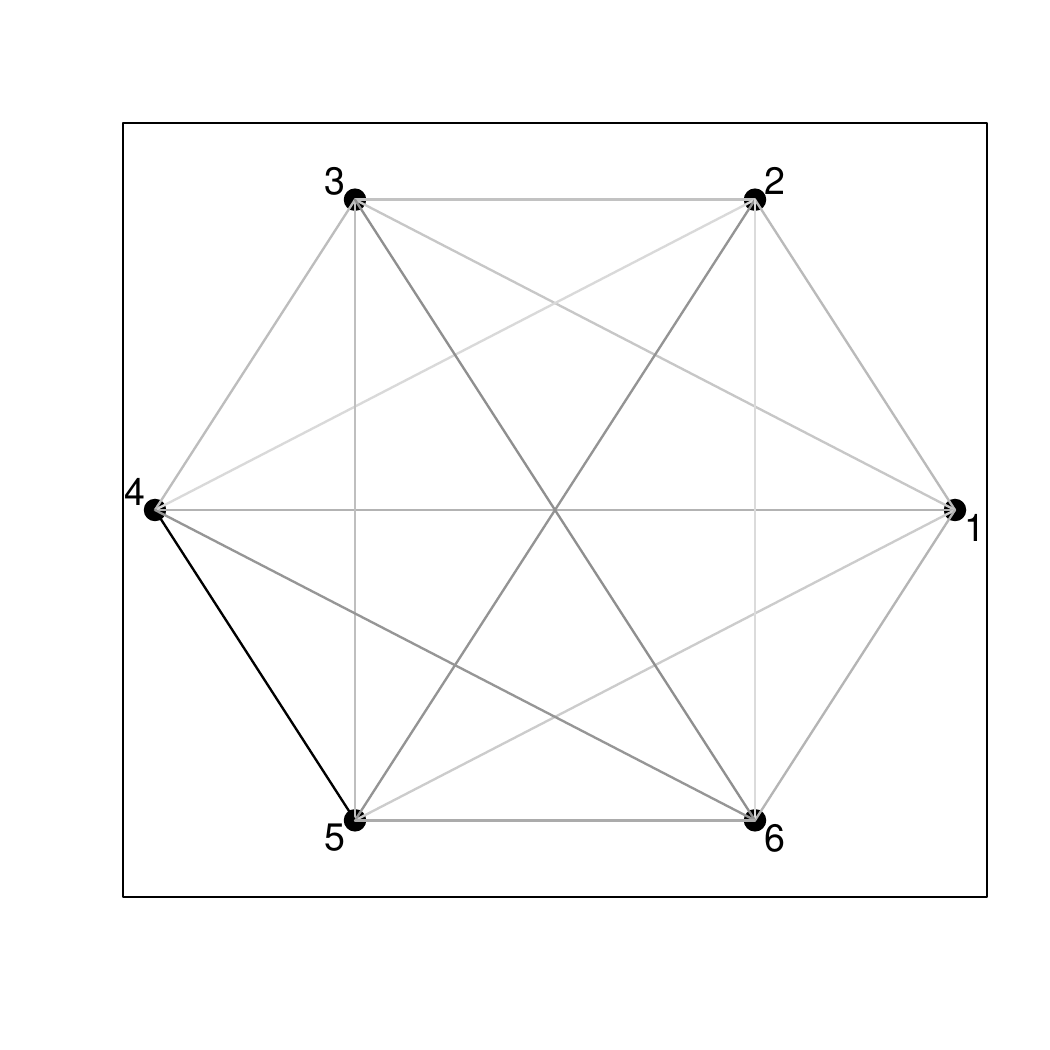}
\caption{Networks in glucose dataset. Gamma model (left column), normal model (right column). Non pregnant women (top row), pregnant women (bottom row).}
\label{fig:glunet}
\end{figure}

\end{document}